\def\year{2022}\relax
\documentclass[letterpaper]{article} 
\usepackage{aaai22}  
\usepackage{times}  
\usepackage{helvet}  
\usepackage{courier}  
\usepackage[hyphens]{url}  
\usepackage{graphicx} 
\usepackage[english]{babel}
\usepackage{amsthm,amsfonts}
\usepackage{hyperref}
\newtheorem{theorem}{Theorem}
\newtheorem{example}{Example}
\newtheorem{definition}{Definition}
\newtheorem{lemma}{Lemma}

\usepackage{natbib}  
\usepackage{caption} 
\DeclareCaptionStyle{ruled}{labelfont=normalfont,labelsep=colon,strut=off} 
\usepackage{algorithm}
\usepackage{newfloat}
\usepackage{listings}
\usepackage{subcaption}
\floatstyle{ruled}
\newfloat{listing}{tb}{lst}{}
\floatname{listing}{Listing}

\title{An Algorithmic Introduction to Savings Circles}
\author {
    Rediet Abebe,\textsuperscript{\rm 1}
    Adam Eck,\textsuperscript{\rm 2}
    Christian Ikeokwu, \textsuperscript{\rm 1}
    Samuel Taggart \textsuperscript{\rm 2}
}
\affiliations {
    \textsuperscript{\rm 1} University of California, Berkeley\\
    \textsuperscript{\rm 3} Oberlin College\\

}

\usepackage{bibentry}

\usepackage[normalem]{ulem}

\def\opt{\ensuremath{\textsc{OPT}}}
\def\alg{\ensuremath{\textsc{EQ}}}

\renewcommand{\labelenumi}{\bf \alph{enumi}.}
\usepackage{mathtools}
\usepackage[dvipsnames]{xcolor}

\newtheorem{corollary}{Corollary}

\newcommand{\vvec}{{\mathbf v}}
\newcommand{\vecv}{{\mathbf v}}
\newcommand{\xvec}{{\mathbf x}}
\newcommand{\vecx}{{\mathbf x}}
\newcommand{\pvec}{{\mathbf p}}
\newcommand{\vecp}{{\mathbf p}}

\DeclareMathOperator*{\argmax}{argmax}
\usepackage[noend]{algpseudocode}
\usepackage{multirow}

\usepackage{thmtools}
\usepackage{thm-restate}

\begin{document}

\maketitle

\begin{abstract}

Rotating savings and credit associations (roscas) are informal financial organizations common in settings where communities have reduced access to formal financial institutions. In a rosca, a fixed group of participants regularly contribute sums of money to a pot. This pot is then allocated periodically using lottery, aftermarket, or auction mechanisms. Roscas are empirically well-studied in economics. They are, however, challenging to study theoretically due to their dynamic nature. Typical economic analyses of roscas stop at coarse ordinal welfare comparisons to other credit allocation mechanisms, leaving much of roscas' ubiquity unexplained. In this work, we take an algorithmic perspective on the study of roscas. Building on techniques from the price of anarchy literature, we present worst-case welfare approximation guarantees. We further experimentally compare the welfare of outcomes as key features of the environment vary. These cardinal welfare analyses further rationalize the prevalence of roscas. We conclude by discussing several other promising avenues.
\end{abstract}

\section{Introduction}\label{sec:intro}

Rotating saving and credit associations (roscas) are financial institutions common in low- and middle-income nations, as well as immigrant and refugee populations around the world. In a rosca, a group of individuals meet regularly for a defined period of time. At each meeting, members contribute a sum of money into a pot, which is then allocated via some mechanism, such as a lottery or an auction. Recipients often use this money to purchase durable goods (e.g., farming equipment, appliances, and vehicles), to buffer shocks (e.g. an unexpected medical expense), or to pay off loans. Roscas often exist outside of legal frameworks and do not typically have a central authority to resolve disputes or enforce compliance. Instead, they provide a decentralized mechanism for peer-to-peer lending, where members who receive the pot earlier borrow from those who receive it later. They also create a structure for mutual support and community empowerment.

Roscas are used in over 85 countries and are especially prevalent in contexts where communities have reduced access to formal financial institutions \cite{aredo2004rotating,bouman1995rosca,klonner2002understanding,la2002inequality,raccanello2009health}. Roscas account for about one-half of Cameroon's national savings. Likewise, over one in six households in Ethiopia's highlands participate in \textit{ekub}, the region's variant of roscas \citet{bouman1995rosca}. Due to their ability to provide quick, targeted support within communities, roscas and other mutual aid organizations often play an instrumental role when communities experience shocks and disasters \cite{chevee2021mutual,mesch2020covid,travlou2020kropotkin}.

Roscas are well-studied in the economics literature, with economic theory on the subject pioneered by \citet{besley1993economics,kovsted1999rotating,kuo1993loans} (see Appendix \ref{app:related} for further related works). This line of work seeks to explain how roscas act as insurance, savings, and lending among members. While such studies have deepened our understanding of roscas, they are typically constrained in two main ways. First, the standard economic approach solves exactly for equilibria, which can be especially difficult due to the dynamic nature of roscas. Second, much of the existing theory focuses on coarse-grained comparisons between the welfare of roscas and other mechanisms for allocating credit. In part due to these coarse comparisons, this work often concludes that roscas allocate credit suboptimally, leaving open the question of why roscas are prevalent in practice. 

In this work, we initiate an algorithmic study of roscas. Viewing roscas through the lens of approximation and using techniques from the price of anarchy literature, we study the welfare properties of rosca outcomes without directly solving for them. 
We specifically quantify the allocative efficacy of roscas: \emph{how well do roscas coordinate saving and lending among participants with heterogeneous investment opportunities}? We show roscas enable a group’s lending and borrowing in a way that approximately maximizes the groups’ total utility. We do so under a wide range of assumptions on both participants’ values for investment and the mechanisms used for allocating the pots. This robustness may provide one explanation for their prevalence.

Our work builds on the saving and lending formulation of \citet{besley1994rotating}. We assume each participant seeks to purchase an investment, such as a durable good, but can only do so upon winning the rosca's pot. We analyze the welfare properties of typical pot allocation mechanisms, such as ``swap roscas,'' where the participants are given an initial (e.g.\ random) allocation and then swap positions in an after-market through bilateral trade agreements. We also study the price of anarchy in auction-based roscas where, during each meeting, participants bid to decide a winner among those who have not yet received a pot. During each round, participants must weigh the value of investing earlier against the utility loss from spending to win that round.  

Our technical contributions are as follows: For swap roscas, we prove that all outcomes guarantee at most a factor $2$ loss. For auction-based roscas, we give full-information price of anarchy results: we study second-price sequential roscas and give a price of anarchy of $3$ under a standard no-overbidding assumption. For first-price sequential roscas, we provide a ratio of $(2e+1)e/(e-1)$. 
Our work provides new applications of and extensions to the {\em smoothness} framework of \citet{syrgkanis2013composable}. However, due to the round-robin (i.e you can only win once) property of rosca allocations and the fact that all payments are redistributed to members, standard smoothness arguments do not immediately yield bounds in our setting. Via a new sequential composition argument, we show that roscas based on smooth mechanisms are themselves smooth, and we go beyond smoothness to bound the distortionary impact of redistributed payments on welfare. Our above results hold under the well-studied assumption of {\em quasilinear} utility for money. We extend most of our theoretical results to nonlinear utility functions, and also use simulations to consider the impact of nonlinear utilities in several natural families of swap roscas.

Overall, this work aims to provide greater exposure for mutual aid organizations more generally and roscas to the algorithms community. In doing so, we present a case study showing how algorithmic game theory can provide a useful perspective for further understanding fundamental questions related to these financial organizations. Given their prevalence and efficacy, insights into roscas can help inform the design of other safety net programs, especially for communities that already commonly use roscas. As new technologies are introduced in low-access contexts, also, the need to understand existing, prevalent financial organizations is even more pressing. We close the paper with discussion of promising research directions.




\section{Model and Preliminaries}

A rosca consists of $n$ participants and takes place over $n$ discrete and fixed time periods, or {\em rounds}. During each round three things occur: (1) each participant contributes an amount $p_0$ into the rosca common pot, (2) a winner for the pot is decided among those who have not yet won, and (3) the winning participant is allocated the entire pot worth $np_0$. Typically, the contributions $p_0$ are decided ex ante during the formation of the rosca. As is common in previous literature, we will not model the selection process for $p_0$, but instead take it as given \citep[c.f.][]{klonner2001roscas}.

With the rosca contribution $p_0$ fixed, we can cast roscas as an abstract multi-round allocation problem, where every participant is allocated exactly one pot, and each pot is allocated to exactly one participant, illustrated in Alg~\ref{pseudo:rosca}.
Each participant's value for the allocations is described by a real-valued vector $\vvec_i=(v_i^1,\ldots, v_i^n)$, with $v_i^t$ representing participant $i$'s value for winning the pot in round $t$ and having access to that money at that time. We denote allocations by $\xvec=(\xvec_1,\ldots,\xvec_n)$, where $\xvec_i=(x_i^1,\ldots,x_i^n)$ is an indicator vector, and $x_i^t=1$ if and only if participant $i$ receives the pot in round $t$. Based on a common observation from previous literature, we can further assume that values for allocation are non-increasing over time: i.e., for $t<t'$ and any $i$, $v_i^t \geq v_i^{t'}$. This follows if rosca funds are used to make lumpy investments, e.g.,\ in a durable good, as is common in practice  \cite{besley1993economics, besley1994rotating, kovsted1999rotating, klonner2008private}. Participants prefer to own the good earlier rather than later, though different participants' values for owning the good earlier may vary.

\algrenewcommand\algorithmicindent{0.4em}
\begin{algorithm}[!ht]
\small
\caption{Rosca Multi-Round Allocation}
\textbf{Constants}: $n$: the number of participants and rounds in the rosca.  $p_0$: amount contributed by each participant to the pot in each round of the rosca.\\
\textbf{Inputs}: Valuations $\vvec$, where $v_i^t$ indicates the value to participant $i$ of winning the pot in round $t$. $Alloc$ an allocation mechanism. \\

For each round $t \in \{1, 2, \ldots, n\}$
\renewcommand{\labelenumi}{\bf \arabic{enumi}.}
\begin{enumerate}
    \item Each participant contributes $p_0$ into the pot
    \item $Alloc$ selects the winning participant (who has not yet won a round) 
    \item The winning participant receives the pot worth $np_0$
    \item \emph{Optional}: Some participants make payments based on $Alloc$, which are redistributed to the others as rebates.
\end{enumerate}
\renewcommand{\labelenumi}{\bf \alph{enumi}.}
\label{pseudo:rosca}
\end{algorithm}

\subsection{Roscas with Payments}

A variety of different pot allocation mechanisms are common in practice \citep[see][]{ardener1964comparative, bouman1995rotating}. This work considers roscas where participants make payments to influence their allocations, and assumes as a first-order approximation that participants are rational.
Payments in roscas take the form $\pvec=(\pvec_1,\ldots, \pvec_n)$, where $\pvec_i=(p_i^1\ldots, p_i^n)$. As participants' abilities to save money over time are typically limited, we assume participants' utilities are additively separable across rounds, but possibly nonlinear in money. That is, participant $i$ with value vector $\vvec_i$ has utility for allocation $\xvec_i$ and payments $\pvec_i$ given by
\begin{equation*}
    u_i^{\vvec_i}(\xvec_i,\pvec_i)=\vecx_i \cdot \vecv_i-\sum\nolimits_tC(p_i^t),
\end{equation*}
for some disutility function $C$ that is both increasing and satisfies $C(0)=0$. In a given round, $p_i^t$ could be positive, if the allocation mechanism requires $i$ to make payments, or negative, if a different participant's payments are redistributed to $i$. We refer to the latter as {\em rebates}, and assume all payments are redistributed each round, i.e.,\ $\sum_i p_i^t=0$ for all $t$.

A participant who makes positive payments in round $t$ has less money to spend in round $t$, and one who receives rebates in the form of negative payments has more to spend. The function $C(\cdot)$ describes participants' preferences for these changes in wealth. A more precise interpretation of $C(p_i^t)$ is as follows: assume that each participant $i$ has a per-round income of $w$. Without participating in the rosca, they would receive a utility $U(w)$ from consumption of that income, for some increasing consumption utility function $U$. Upon contributing $p_0$ into the rosca pot each round, the participant's baseline consumption utility is $U(w-p_0).$ If the rosca's allocation procedure requires additional payments (or distributes rebates) of $p_i^t$, a participant's utility from consumption becomes $U(w-p_0-p_i^t)$. The disutility function $C$ then represents the participant's difference in utility for consumption, $$C(p)=U(w-p_0)-U(w-p_0-p),$$ which is increasing.

A large body of anthropological and empirical work on roscas shows that participants in the same rosca tend to have similar economic circumstances \citep[see][]{ardener1964comparative, aredo2004rotating, mequanent1996role}. So, following the theory literature, we assume $U$ and $w$ (and hence $C$) are identical across participants, even if the value for receiving the pot differ between participants \citep{besley1994rotating, kovsted1999rotating, klonner2001roscas}.  It is typical to assume consumption utility $U$ is weakly concave, and hence $C$ is weakly convex \citep{anderson2002economics, klonner2003rotating, klonner2001roscas}. The special case of {\em quasilinear utilities}, where $C(p)=p$, is especially well-studied in the algorithmic game theory literature.

To measure allocative performance of a rosca, we study the participants’ total utility: $$\textsc{Wel}^{\vecv }(\vecx,\vecp)=\sum\nolimits_iu_i^{\vecv_i}(\vecx_i,\vecp_i).$$
 Following the interpretation of $C$ in terms of consumption utility $U$, $\textsc{Wel}^{\vecv}(\vecx,\vecp)$ represents the gain in utility to all participants for a given allocation $\vecx$ and payments $\vecp$, above the baseline total utility of $n^2U(w-p_0)$, obtained by each of the $n$ participants obtaining utility $U(w-p_0)$ for $n$ rounds. Among all possible matchings $\vecx$ and payment profiles $\vecp \in \mathbb{R}^{n \times n}$, the optimal welfare-outcome is given by the maximum-weight matching ${\vecx}^*=\argmax_{\vecx} \sum_i{\vecx_i \cdot \vecv_i}$ and payments ${\vecp}^*=(0,\ldots,0)$, whose welfare is denoted $\textsc{OPT}(\vecv)=\textsc{Wel}^{\vecv}({\vecx}^*, \vecp^*)$.

To quantify the inefficiency of a rosca outcome $(\vecx,\vecp)$, we study the approximation ratio $\textsc{OPT}(\vecv)/\textsc{Wel}^\vecv(\vecx,\vecp)$. When rosca outcomes are equilibria of auctions, as in Section~\ref{sec:auction}, this ratio is also known as the price of anarchy (PoA).

\vspace{2mm}
\noindent \textbf{Roadmap.} The remainder of this paper proceeds as follows: In Section~\ref{sec:auction}, we prove a constant-approximation for auction roscas. We do the same for swap roscas in Section~\ref{sec:swap}. Both sets of results focus on quasilinear utilities, where $C(p)=p$, and hence all welfare loss comes from allocative inefficiency. We extend these results to nonlinear utilities in the supplement. In Section~\ref{sec:swap}, we further conduct experiments to study the impact of nonlinear utility on swap rosca welfare. We give directions for future work in Section~\ref{sec:conclusion}.

\section{Auction Roscas}
\label{sec:auction}
\newcommand{\vecb}{{\mathbf b}}
\newcommand{\veca}{{\mathbf a}}

Auctions are a common mechanism for allocating pots in roscas \cite{ardener1964comparative, bouman1995rosca, klonner2003buying}. Two major sources of variety in auction roscas are (1) when bids are solicited from participants and (2) the type of auction run for the bidding process. The bidding may occur either at the beginning, in which case a single (up-front) auction determines the full schedule of pot allocations, or sequentially, in which case a separate auction is held each period to determine the allocation for the corresponding pot. We consider sequential first- and second-price (equivalently, ascending- and descending-price) auctions, as well as up-front all-pay-style auctions. Payments are typically redistributed as rebates among all of the non-winning participants.

The fact that outcomes depend on participants' bidding behavior complicates our analysis. We assume participants play a {\em Nash equilibrium (NE)} of the rosca's auction game. That is, their bidding strategy maximizes their utility given the bidding strategies of other participants. Our analysis will use the {\em smoothness} framework of \citet{syrgkanis2013composable}, along with new arguments to handle rosca-specific obstacles. We assume participants have quasilinear utilities. 

\subsection{Proof Template: Up-Front Roscas}
\label{sec:upfront}

We begin our analysis by considering roscas with up-front bidding. In an up-front 
rosca, each participant $i$ submits a bid $b_i$ at the beginning of the rosca. Participants pay their bids, and are then assigned pots in decreasing order of their bids, with the highest participant receiving the pot in round $1$, and so on. Each participant $i$'s payments are redistributed evenly among the other participants in the form of reduced per-period payments into the rosca. Under quasilinear utilities, it is not relevant to the participants' utilities what round payments are made; the only relevant outcome is total payments, which we write as $p_i=\sum_tp_i^t$ when context allows. We can further assume per-period payments remain fixed and that the participants receive the redistributed payments up-front in the form of a rebate. We decompose the participants' total payments into their {\em gross payments} $\hat p_i$ and {\em rebates} $\hat r_i$, with $p_i=\hat p_i-\hat r_i$.  Formally:  
\begin{definition}
In an {\em up-front rosca} with quasilinear participants, each participant $i$ submits a bid $b_i$, with $\vecb=(b_1,\ldots, b_n)$. Let $r_i$ denote the rank of participant $i$'s bid. Allocations are $x_i^{t}(\vecb)=1$ if $t=r_i$ and $0$ otherwise. Participant $i$'s gross payment is $\hat p_i(\vecb)=b_i$, and their rebate is $\hat r_i(\vecb)=\sum_{i'\neq i}b_{i'}/(n-1)$.
\end{definition}

Our auction rosca analyses all follow from a two-step argument. First, we use or modify the smoothness framework of \citet{syrgkanis2013composable} to obtain a tradeoff between participants' utilities and their gross payments. Without rebates, typical auction analyses conclude by noting that high payments imply high welfare. However, because gross payments in roscas are redistributed, it could happen that both gross payments and rebates are high, but welfare is low. Our second step is to rule out this problem. For up-front roscas, we can demonstrate both steps simply.

The first step follows from Lemma A.20 of \citet{syrgkanis2013composable}:

\begin{lemma}\label{lem:upfront}
    With quasilinear participants, any Nash equilibrium $\vecb$ of any up-front rosca with values $\vecv$ satisfies
    \begin{equation}\label{eq:upfrontsmooth}
        \sum\nolimits_i u_i^{\vecv_i}(\vecb)\geq\tfrac{1}{2}\textsc{OPT}(\vecv)-\sum\nolimits_i\hat p_i(\vecb).
    \end{equation}
\end{lemma}

The left hand side of (\ref{eq:upfrontsmooth}) is the equilibrium welfare. It therefore suffices for the second step to upper bound the gross payments on the right hand side.

\begin{lemma}\label{lem:upfrontoverbid}
    Let $\vecb$ be a Nash equilibrium of an up-front rosca with quasilinear participants and values $\vecv$. Then, for any participant $i$, $\hat p_i(\vecb)\leq \vecv_i\cdot\vecx_i(\vecb)$.
\end{lemma}
\begin{proof}
Assume for some $i$ that $\hat p_i(\vecb)> \vecv_i\cdot\vecx_i(\vecb)$. Then participant $i$ must be overbidding. They could improve their utility by bidding $0$, which, in an up-front rosca, does not change their rebates: $\hat r_i(0,\vecb_{-i})=\hat r_i(\vecb)> \vecv_i\cdot \vecx_i(\vecb)-\hat p_i(\vecb)+\hat r_i(\vecb)$.
\end{proof}

Since $\sum_i\vecv_i\cdot\vecx_i(\vecb)$ is equal to the equilibrium welfare, Lemmas~\ref{lem:upfront} and \ref{lem:upfrontoverbid} together imply the following.

\begin{theorem}\label{thm:upfront}
With quasilinear participants, every Nash equilibrium of an up-front rosca has PoA at most $4$.
\end{theorem}

\subsection{Sequential Roscas}
\label{sec:sequential}

We now consider roscas with separate sequentially-held first- or second-price auctions for each pot as opposed to the single-auction format from the previous section.

\begin{definition}
A {\em first-price rosca} runs a first-price auction in each round.

That is, if the highest-bidding participant in round $t$ among those who have not yet won is participant $i^*$, with bid $b_{i^*}^t$, then $x_{i^*}^t=1$, $x_i^t=0$ for all other participants $i$. The gross payments are $\hat p_{i^*}^t=b_{i^*}^t$ and $\hat p_i^t=0$ otherwise. The rebates are $\hat r_i^t=b_{i^*}^t/(n-1)$ for all $i\neq i^*$.
\end{definition}

\begin{definition}
A {\em second-price rosca} runs a second-price auction in each round.
That is, if the highest-bidding participant in round $t$ among those who have not yet won is participant $i^*$, with second-highest bid $b_{(2)}^t$, then $x_{i^*}^t=1$, $x_i^t=0$ for all other participants $i$. The gross payments are $\hat p_{i^*}^t=b_{(2)}^t$ and $\hat p_i^t=0$ otherwise. The rebates are $\hat r_i^t=b_{(2)}^t/(n-1)$ for all $i\neq i^*$.
\end{definition}

Sequential auctions require a monitoring scheme in which the auctioneer discloses information about participants' bids after each round. Our results will hold for any deterministic monitoring scheme. A key subtlety is that participants' actions are now {\em behavioral strategies}: that is, at each stage, participants observe the disclosed history of play so far and can condition their future bids on this history. We denote the vector of behavioral strategies by $\veca=(a_1,\ldots,a_n)$, and denote by $b_i^t$ participant $i$'s bid in a round $t$.

As with up-front roscas, we first derive a tradeoff between utility and gross payments, and second consider the impact of rebates. The sequential format complicates both steps. Our first step will follow from a novel composition argument, where we show that both first- and second-price roscas inherit a tradeoff from their single-item analogs. For second-price roscas, a standard no-overbidding assumption then bounds the auction's rebates and implies a welfare bound. For first-price roscas, we give a more involved analysis that bounds overbidding and yields an unconditional guarantee. Overbidding can both occur in equilibrium and harm welfare, so such an analysis is necessary.

Observe that, first- and second-price roscas can be thought of as the sequential composition $n$ single-item auctions, with a rule excluding past winners. Formally:

\begin{definition}[Round-Robin Composition]
Given a single-item auction $M$, the {\em$n$-item round-robin composition} of $M$ is a multi-round allocation mechanism for $n$ items using the following procedure: During each round $t$, each participant $i$ who has not yet been allocated an item submits a bid $b_i^t$. The mechanism then runs $M$ among the remaining participants to determine the allocation and payments for that round.
\end{definition}

The following definition of {\em smoothness}, adapted from \citet{syrgkanis2013composable}, lets us characterize both first- and second-price roscas with the same framework. For our purposes, it applies to any auction where in round $t$, each bidder who has not yet won submits a real-valued bid $b_i^t$, which we term {\em sequential single-bid auctions}. Note that this includes single-item auctions. We will show that smoothness of single-item auctions implies smoothness of their round-robin composition.

\begin{definition}
Let $M$ be a sequential single-bid auction. We say $M$ is $(\lambda,\mu_1,\mu_2)$-smooth if for every value profile $\vecv$ and action profile $\veca$, there exists a randomized action $a_i^*(a_i,\vecv)$ for each $i$ such that:
\begin{align*}
    \sum\nolimits_i(\vecv_i\cdot &\vecx_i(a_i^*(a_i,\vecv),\veca_{-i})-p_i(a_i^*(a_i,\vecv),\veca_{-i}) \geq \\ &\lambda \textsc{OPT}(\vecv)-\mu_1\sum\nolimits_i \hat p_i(\veca)-\mu_2\sum\nolimits_i B_i(\veca),
\end{align*}
where $B_i(\veca)$ is $i$'s bid in the round where they win, or $0$ if no such round exists.
\end{definition}

\citet{syrgkanis2013composable} show that single-item first-price and second-price auctions are $(1-1/e,1,0)$-smooth, and $(1,0,1)$-smooth, respectively. However, the smoothness result they prove for a form of sequential composition fails to hold for round-robin composition, due to the cardinality constraint on allocations as in our setting. Here, we instead give a new composition argument tailored specifically to the rosca setting, that relies on values decreasing in time. Our composition result follows the following useful definition: 

\begin{definition}
    A single-item mechanism $M$ with allocation rule $\vecx$ and payments $\vecp$ is {\em strongly individually rational (IR)} if (1) for every profile of actions $\veca$, $x_i(\veca)=0$ only if $\hat p_i(\veca)=0$, and (2) there exists an action $\bot$ such that for all $i$ and $\veca_{-i}$, $\hat p_i(\bot,\veca_{-i})=0$.
\end{definition}

\begin{restatable}{lemma}{lemroundrobin}\label{lem:roundrobin}
    Let $M$ be a strongly individually-rational single-item mechanism. If $M$ is $(\lambda,\mu_1,\mu_2)$-smooth for $\lambda\leq1$ and $\mu_1,\mu_2\geq 0$, then its round-robin composition is $(\lambda,\mu_1+1,\mu_2)$-smooth as long as $v_i^t\geq v_i^{t+1}$ for all $i$ and $t$.
\end{restatable}

Our proof of this lemma, presented in the supplement, augments the main idea from the \citet{syrgkanis2013composable} composition result with ideas from \citet{KKT15}, who consider smoothness of non-sequential mechanisms for cardinality-constrained allocation environments. As a corollary of Lemma~\ref{lem:roundrobin}, we obtain that first- and second-price roscas are respectively $(1-1/e,2,0)$ and $(1,1,1)$-smooth.

We next analyze the impact of rebates. If we assume no participant overbids, then payments (and hence rebates) are necessarily bounded by values, and we obtain a similar conclusion to Lemma~\ref{lem:upfrontoverbid}. Moreover, we show in the supplement that an overbidding assumption is necessary for second-price roscas, as is often the case for auctions with second-price payments. The overbidding assumption we require is as follows:

\begin{definition}\label{def:nob}
    Action profile $\veca$ satisfies {\em no-overbidding} if $B_i(\veca)\leq \vecv_i\cdot \vecx_i(\veca)$ for every participant $i$.
\end{definition}

\begin{theorem}\label{thm:smoothrosca}
    Let $M$ be a strongly IR, single-item auction that is $(\lambda,\mu_1,\mu_2)$-smooth, with $\lambda\leq1$. With quasilinear participants, every no-overbidding Nash equilibrium of the corresponding auction rosca with rebates has PoA at most $(2+\mu_1+\mu_2)/\lambda$.
\end{theorem}

\begin{proof}
    Lemma~\ref{lem:roundrobin} implies that the rosca is $(\lambda,1+\mu_1,\mu_2)$-smooth before rebates. We can therefore write:
    \begin{align*}
        \sum\nolimits_i u_i^{\vecv_i}(\veca)&\geq \sum\nolimits_iu_i^{\vecv_i}(a_i^*,\veca_{-i})\\
        &\geq \sum\nolimits_i( \vecv_i\cdot \vecx_i(a_i^*,\veca_{-i})-\hat p_i(a_i^*,\veca_{-i}))\\
        &\geq \lambda\textsc{OPT}(\vecv)-(1+\mu_1)\sum\nolimits_i\hat p_i(\veca) \\
        &~~~~~~~~~~~~~~~~~~~~~-\mu_2\sum\nolimits_iB_i(\veca)\\
        &\geq \lambda\textsc{OPT}(\vecv)-(1+\mu_1+\mu_2)\sum\nolimits_iB_i(\veca)\\
        &\geq \lambda\textsc{OPT}(\vecv)-(1+\mu_1+\mu_2)\sum\nolimits_i\vecv_i\cdot \vecx_i(\veca)
    \end{align*}
    Since both $\sum_i u_i^{\vecv_i}(\veca)$ and $\sum_i \vecv_i\cdot \vecx_i(\veca)$ are equal to equilibrium welfare, the result follows.
\end{proof}

\begin{corollary}
    For quasilinear participants, any Nash equilibrium of the first-price rosca satisfying no-overbidding has PoA at most $3e/(e-1)$.
\end{corollary}

\begin{corollary}
    For quasilinear participants, any Nash equilibrium of the second-price rosca satisfying no-overbidding has PoA at most $3$.
\end{corollary}

\subsection{Relaxing No-Overbidding}
\label{sec:overbidding}

The no-overbidding assumption in the previous section rules out behavior where participants overbid in early rounds to induce others to bid high in later rounds, thereby resulting in high rebates. When this behavior is extreme, participants' payments could conceivably far exceed their values, which in turn complicates the smoothness-based approach. The following example gives a Nash equilibrium of a first-price rosca where overbidding leads to welfare loss.

\begin{example}\label{ex:loss}
    Consider three participants, with $\vecv_1 = (1,0,0)$, $\vecv_2=(2,2,0)$, and $\vecv_3=(2,2,0)$. The following behavioral strategies form a Nash equilibrium. Participant 1 bids 2 in round 1. Participants 2 and 3 bid 1 in round 1. If participant 1 bids less than 2 in round 1, participants 2 and 3 bid 0 in round 2. Otherwise, they bid 2. The optimal welfare is then $4$, but the equilibrium welfare is $3$.\footnote{This example does not satisfy the refinement of subgame perfection, though our welfare guarantees do not need this restriction.}
\end{example}

Despite the loss exhibited in Example~\ref{ex:loss}, we can obtain a constant price of anarchy for first-price roscas without an overbidding assumption. Lemma~\ref{lem:overpay} below shows that overbidding cannot drive payments much higher than equilibrium welfare. The lemma extends the following logic: In equilibrium, the participant who wins in the final round has no competition, and is therefore making zero payments. Consequently, the participant who wins in the second-to-last round cannot expect any rebates from round $n$, and therefore has no incentive to overbid. This, in turn, limits the rebates due the participant who wins the round before that, and so on. These limits on rebates limit the extent of overbidding that might occur. Throughout this section, we index participants such that in round $t$, the winner is participant $t$.

\begin{restatable}{lemma}{lemoverpay}\label{lem:overpay}
    Fix a Nash equilibrium of a first-price rosca. Then:
    $$\hat p_{t}^t\leq v_{t}^t+\tfrac{1}{n-1}\sum_{t'=t+1}^nv_{t'}^{t'}(\tfrac{n}{n-1})^{t'-t-1}.$$
\end{restatable}

We provide the proof in the supplementary materials.

\begin{theorem}\label{thm:firstprice}
    In any Nash equilibrium of the first-price rosca, the PoA is at most $(2e+1)e/(e-1)$.
\end{theorem}

The result follows from summing the bounds on $\hat p_i(\veca)$ from Lemma~\ref{lem:overpay}, which can be arranged to obtain an upper bound of $e\sum_i \vecv_i\cdot \vecx_i(\veca)$ of the total gross payments. The theorem then follows from applying smoothness as before.

\subsection{Extension to Nonlinear Utilities}

In Appendix~\ref{app:nonlinear}, we extend the price of anarchy results above beyond quasilinear utilities. With arbitrary convex cost for payments $C$, the setting comes to resemble hard budgets, for which the price of anarchy is known to be poor. We parametrize our results by upper ($\beta$) and lower ($\alpha$) bounds on the slope $C'$. We give performance guarantees which scale linearly with the ratio $\beta/\alpha$. For up-front roscas, our bounds are unconditional, while for sequential roscas, we assume an analogous no-overbidding condition to the quasilinear version.

\section{Swap Roscas}
\label{sec:swap}

Several common rosca formats eschew competition between participants in allocating pots. Examples include roscas based on random lottery allocations or those based on seniority or social status \cite{anderson2009enforcement, kovsted1999rotating}. To improve total welfare, it is common practice for participants to engage in an \emph{aftermarket} by buying or selling their assigned allocations when it is mutually beneficial, i.e., by swapping rounds in the rosca \cite{mequanent1996role}.

In this section, we formally define these swap roscas and show that, for participants with quasilinear utilities ($C(p)=p$), this aftermarket is guaranteed to converge to an outcome that yields at least half of the optimal welfare. We then present experimental results showing that this guarantee is often better, even for strictly convex $C$.

\subsection{Theoretical Analysis}

As is common in the literature and in practice, we assume that the aftermarket occurs via a series of two-agent swaps \cite{mequanent1996role, bouman1995rotating, ardener1964comparative}. We assume these swaps can occur at any round $t$. We denote by $\vecp^t$ the vector of payments for round $t$, which are initialized to $0$ for each round and updated as swaps occur. A swap occurs if and only if it is utility-improving for two participants under some set of payments. Formally:
\begin{definition}
Given initial allocation $\vecx$ and payments $\vecp^t$ at round $t$, a {\em swap} is given by a pair of participants $i$, $i'$ assigned to rounds $j,j'\geq t$, respectively, and a payment $\hat p$. A swap is valid if $v_i^{j'}-C(p_i^t+\hat p)>v_i^j-C(p_i^t)$ and $v_{i'}^j-C(p_{i'}^t-\hat p)>v_{i'}^{j'}-C(p_{i'}^t).$

Upon executing a swap, set $x_i^j,x_{i'}^{j'}\leftarrow 0$, $x_i^{j'},x_{i'}^j\leftarrow 1$, $p_i^t\leftarrow p_i^t+\hat p$, and $p_{i'}^t\leftarrow p_i^t-\hat p$.
\end{definition}

Note that with quasilinear participants, all valid swaps must strictly improve allocative efficiency since $C(p)=p$. That is, $v_i^{j'}+v_{i'}^j>v_i^j+v_{i'}^{j'}$, and the validity of a swap does not depend on the initial payments $\vecp^t$. We then study roscas of the following form:

\begin{definition}\label{def:swap}
A {\em swap rosca} starts from an initial allocation $\vecx$ and initial payments $\vecp=\{\vecp^t\}_{t=1}^n$ of $0$ for each participant and round. At each round $t=1,\ldots, n$, participants execute valid swaps and we update the allocation and payment accordingly. We do so until there are no valid swaps.
\end{definition}

Note that for non-linear $C$, new swaps may become valid moving from round $t$ to $t+1$, as each new round's payments reset to $0$. For quasilinear participants, however, Definition~\ref{def:swap} executes all swaps in round $1$. In this case, the resulting allocation is guaranteed to be stable to pairwise swaps.

\begin{definition}
An allocation $\vecx$ is {\em swap-stable} if for all participants $i,i'$ assigned to $j,j'$, we have that
$v_i^j+v_{i'}^{j'}\geq v_i^{j'}+v_{i'}^j$.
\end{definition}

For quasilinear participants, swap-stability is guaranteed regardless of the initial allocation. Convergence of the swap process follows from the fact that the total allocated value $\sum_i\vecv_i\cdot \vecx_i$ strictly increases each swap and that the number of allocations is finite.

\begin{theorem}\label{thm:swap}
For quasilinear participants, the welfare approximation for every swap rosca is at most $2$.
\end{theorem}

\begin{proof}
Without loss of generality, assume that the welfare-optimal allocation assigns each participant $i$ to be allocated the pot in round $i$, so the optimal welfare is $\sum_i v_i^i$. Now let $\pi(i)$ denote the \emph{round} when participant $i$ is allocated the pot in the swap rosca's final allocation, and $\pi^{-1}(i)$ the \emph{participant} allocated the pot in round $i$. Note that $\pi$ and $\pi^{-1}$ are bijections. Furthermore, under quasilinear utilities, all payments between participants are welfare-neutral, and hence the rosca welfare is given by $\sum_iv_i^{\pi(i)}$.

For any participant $i$, note that swap-stability implies $$v_i^{\pi(i)}+v_{\pi^{-1}(i)}^i\geq v_i^i+v_{\pi^{-1}(i)}^{\pi(i)}\geq v_i^i.$$ Summing over all participants $i$, we get $$\sum\nolimits_iv_i^{\pi(i)}+\sum\nolimits_iv_{\pi^{-1}(i)}^i\geq \sum\nolimits_i v_i^i.$$ Since $\pi$ and $\pi^{-1}$ are bijections, both sums on the lefthand side are equal to the rosca welfare, and the righthand side is the optimal welfare, giving us a $2$-approximation.
\end{proof}

Example 1 in the appendix shows that this bound is tight.

\subsection{Experimental Results}
\label{sec:experiment}

The results presented so far partially rationalize the prevalence of auction and swap roscas. However, two limitations prevent a comprehensive view of roscas' allocative efficiency. First, the worst-case nature of our theoretical results give little detail about outcomes in \emph{typical} instances. Second, our results hold only under quasilinear utilities, which may be less realistic for extremely vulnerable participants.

This section complements our theoretical results with computational experiments that shed light on these latter questions for swap roscas. We simulate swap roscas under natural instantiations of participants' values, and with participants' costs for payments taking a well-studied but non-linear form. We find that the approximation ratio of these roscas in more typical scenarios is significantly better than the worst-case ratio, even after relaxing quasilinearity. Our experiments also allow us to study the way rosca performance changes as participants' values for their payments become more convex. In particular, we use
{\em constant relative risk aversion (CRRA)} utilities, given by
\begin{equation*}
    C(p;W,a)=(1-a)^{-1}(W^{1-a}-(W-p)^{1-a}),
\end{equation*}
where the parameter $W$ represents the participant's starting wealth, and $a$ governs the convexity of the function, with $a=0$ being quasilinear. For $a>0$, CRRA utilities have a vertical asymptote at $p=W$, as participants are unable to spend beyond their means. We choose $W$ to be less than many of our participants' maximum values for the rosca pot. This is intended to capture that most participants cannot afford the durable good without the rosca \cite{anderson2002economics}. Note that as $a\rightarrow1$, $C(p;W,a)\rightarrow \ln(W)-\ln(W-p).$ We choose CRRA utilities because they are standard for modeling preferences for wealth in economics \citep[see, e.g.][]{R96}.

We give two sets of experimental results. In each, we run $9$- and $30$-person roscas (typical sizes for small- and medium-sized roscas), and compare three quantities: the optimal welfare under our selected value profile, the expected approximation ratio of a random allocation before any swaps, and the approximation ratio for a swap rosca run from a random allocation. Our swap roscas are simulated according to the description in Section~\ref{sec:swap}. For a pair of participants $i$ and $j$ for whom there exists a valid swap, there are generally many payments which will incentivize a swap and we choose the smallest such payment.

\subsection{Experiment: CRRA Utilites}
\label{sec:expcrra}

Our first experiment fixes a profile of participant values and studies the performance of swap roscas as the convexity parameter $a$ and starting wealth $W$ vary. The value profile, comprised of $9$ participants, features $6$ with {\em cutoff values} of the form $v_i^t = \overline v$ for all $i\leq \hat t$ for some $\hat t$, and three participants with values which are roughly linearly decreasing in time. The average maximum value among cutoff participants is $5$, which matched the average value for linearly decreasing values. We give all value profiles explicitly in the supplement. We consider values of $a$ ranging from $0$ (quasilinear) to $2$ (very convex), focusing on smaller values, as larger values of $a$ tend to represent very similar, extreme functions. We take $W$ in the range $\{1,\ldots, 5\}$, as this puts participants' wealth levels generally below their values for the rosca pot. Welfare values are averaged over $10,000$ simulation runs, each starting with a random initial allocation that participants can pay to improve through swaps. Results for this simulation can be found in Table~\ref{table:crra9}.

\begin{table}[!ht]
\centering
\caption{Swap Rosca Performance Under Different CRRA Parameters (OPT = 45, random baseline ratio = 1.601)}
\label{table:crra9}
\begin{tabular}{cc|c|c|c|c|c}
 & & \multicolumn{5}{c}{$\mathbf{W}$} \\
 & & $\mathbf{1}$ & $\mathbf{2}$ & $\mathbf{3}$& $\mathbf{4}$& $\mathbf{5}$ \\ \cline{2-7}
 \multirow{9}{*}{$\mathbf{a}$} & $\mathbf{0}$ & 1.035 & 1.034 & 1.035 & 1.034 & 1.035 \\ \cline{2-7}
 & $\mathbf{0.1}$ & 1.121 & 1.119 & 1.070 & 1.067 & 1.063 \\ \cline{2-7}
 & $\mathbf{0.2}$ & 1.122 & 1.121 & 1.080 & 1.074 & 1.074 \\ \cline{2-7}
 & $\mathbf{0.3}$ & 1.122 & 1.119 & 1.118 & 1.086 & 1.081 \\ \cline{2-7}
 & $\mathbf{0.5}$ & 1.122 & 1.124 & 1.121 & 1.119 & 1.120 \\ \cline{2-7}
 & $\mathbf{0.75}$ & 1.122 & 1.122 & 1.123 & 1.121 & 1.121 \\ \cline{2-7}
 & $\mathbf{1}$ & 1.124 & 1.122 & 1.121 & 1.123 & 1.122 \\ \cline{2-7}
 & $\mathbf{1.5}$ & 1.123 & 1.123 & 1.122 & 1.122 & 1.123 \\ \cline{2-7}
 & $\mathbf{2}$ & 1.122 & 1.125 & 1.123 & 1.122 & 1.123 \\
\end{tabular}
\end{table}

Across all values of $W$, the approximation ratio of swap roscas generally worsens (increases) as the level of convexity $a$ increases. Intuitively, this is likely due to the fact that since $C$ is convex, a participant receiving payments for a swap values them less than the participant offering the payments. Consequently, swaps are less likely to occur, even if they would lead to improved allocative efficiency. Meanwhile, the effect of $W$ depends on the level of convexity $a$. When $0 < a < 0.5$, participants with higher wealth $W$ have more money to spend on swaps, making swaps more likely to occur and hence improve allocative efficiency.  Thus, approximation ratios improve (decrease) with higher $W$.  However, as convexity increases, the disincentive to swap caused by convexity overcomes the benefit of having greater wealth with which to pay for swaps, and the approximation ratios no longer change with $W$. For all parameter values chosen, however, swap roscas led to a marked improvement over the approximation ratio from random allocation alone, suggesting that even under extreme convexity, participants are able to identify local improvements to social welfare.  We also repeat this experiment with a 30-participant rosca using similar value profiles and observe the same trends. We present the results in the supplement.

\subsection{Experiment: Distributional Diversity}
\label{sec:expneeds}

Our second set of experiments, discussed in more detail in the supplement, varies the distribution of values across the population of participants, again for $9$- and $30$-person roscas. This allows us to study the way the distribution of need across a population impacts rosca welfare. We find that performance is insensitive to wide inequality in values of participants in the population.

\section{Discussion and Conclusion}
\label{sec:conclusion}

Roscas are complex and varied social institutions, significant for their integral role in allocating financial resources worldwide. In this work, we focus specifically on the allocative efficacy of roscas as lending and saving mechanisms. We derive welfare guarantees for roscas under a variety of allocation protocols and show that many commonly-observed roscas provide a constant factor welfare approximation to the optimal allocation. This guarantee, we believe, gives partial explanation for the ubiquity of roscas. In addition to these specific results, our work also serves as proof of concept for the potential for techniques from algorithmic game theory to help us better understand roscas and, more generally, how communities self-organize to create opportunity. We highlight ideas for further exploration below. 

First, our work modeled the savings aspect of roscas, though roscas are also used as insurance when participants experiencing unanticipated needs may bid to obtain the pot earlier than they may have otherwise planned \cite{calomiris1998role,klonner2003rotating,klonner2001roscas}. There remain many gaps in our understanding of roscas when participants' values and incomes evolve stochastically over time.

Another challenge is understanding the tension between allocative efficiency and wealth inequality. Participants with valuable investment opportunities might not bid as aggressively if their low wealth causes them to value cash highly. This is exacerbated when participants experience income shocks, which is often experienced by economically vulnerable individuals \cite{abebe2020subsidy,nokhiz2021precarity}. Ethnographic work shows that altruism plays a significant role in alleviating this tension \cite{klonner2008private, sedai2021friends}. Roscas often serve a dual role of community-building institutions. Consequently, participants tend to observe signals about each others’ shocks, and act with mutual aid in mind \cite{klonner2008private, mequanent1996role}.

Though roscas often work outside formal institutions, studies show that ``rosca enforcement" is not often an issue. For instance, \cite{smets2000roscas,van1997microeconomics} show that early recipients of the pot rarely default, in part due to strong community norms and standards. These considerations often go unaccounted for in theoretical studies of roscas. A deeper understanding of community norms and standards can shed more light on rosca enforcement mechanisms and robustness.  

Finally, there are many questions on how aspects of the population and environment govern the performance of roscas: i.e., under what conditions would one prefer one type of rosca over another? Similarly, how do roscas perform when their members evolve over time, e.g., with some participants joining part way through the rosca and potentially holding more leverage? Likewise, rosca formation is known to be crucial, with many roscas preferring individuals with similar socio-economic backgrounds. Modeling and examining the rosca formation process can improve our understanding of the interaction between the rosca formation process and their functionality, efficacy, and robustness.

\bibliography{references}

\newpage 

\newpage 
\textcolor{white}{hello}
\newpage 
\appendix 

\section{Further Related Works}
\label{app:related}
Roscas are well-documented as both pervasive and effective in promoting positive economic, social, and even health outcomes. Beyond works already mentioned, \citet{raccanello2009health} document the use of roscas to finance healthcare expenditures and build wealth in Mexico. \citet{aredo2004rotating} demonstrates the flexible and varied nature of roscas in Ethiopia.
\citet{pasha2016role} show that ekub are an engine of small business finance in the city of Arba Minch, and that private businesses actually prefer raising money from roscas than from formal financial institutions. \citet{amankwah2019pareto} and \citet{alabi2007role} study roscas in Ghana, \citet{ogujiuba2013challenges} in Nigeria, and \citet{kabuya2015rotating} in Eswatini. \citet{alabi2007role} show from evidence in Ghana that people joined roscas for their perceived efficiency, and that roscas facilitated small-scale business enterprises.

Many studies analyze composition and participation differences across age, ethic, gender, and socioeconomic lines. \citet{adams1992rotating} and \citet{ardener1995money} show that rosca participation is higher among women than men. \citet{anderson2002economics} shows how employed married women in particular in Kenya use roscas to save, protecting earnings from their husbands' more immediate-minded spending. Roscas are known to often include members from similar socio-economic backgrounds \cite{aredo2004rotating}. Nonetheless, \citet{klonner2008private} shows that intragroup diversity is associated with higher rates of bidder altruism and more efficient intra-rosca allocations. 

Economists have also studied the interaction of roscas with formal credit markets. For instance, \citet{besley1994rotating} show that while credit markets are more efficient than roscas, there are situations in which one can expect a higher ex ante expected utility in roscas than formal credit markets. Relatedly, \citet{fang2015rosca} show that in cases where formal credit markets are present but imperfect, roscas and credit markets can complement one another, thereby improving social welfare. 

A different line of work studies roscas as a form of insurance. \citet{klonner2001roscas} develops the first such model of roscas, comparing their performance to risk-sharing contracts.
In this model, roscas can serve as a financial intermediary and benefit risk-averse participants. This work also analyzes the risk-sharing performance of several bidding roscas run simultaneously. The conclusion is that this set-up matches the performance of linear risk-sharing contracts while boasting greater enforceability. Other studies of roscas as insurance include \citet{baland2019now}, and \citet{calomiris1998role}.

This work is most closely related existing analyses of roscas' efficiency. \citet{besley1993economics,besley1994rotating} introduce the first theoretical model of roscas, and a strong focus these and of subsequent studies is on providing comparative welfare guarantees, e.g.\ between different types of roscas, or between roscas and alternative financial institutions. These results typically require strong assumptions, e.g.\ on homogeneity values of either across participants or over time. For example,  \citet{besley1993economics,besley1994rotating} show that both the random and bidding rosca are inefficient, but do not give bounds on this inefficiency. \citet{kovsted1999rotating} analyze differences between random and bidding roscas, again under the assumption that people are saving for a large purchase. They allow for some heterogeneity in people's access to credit, and again provide a comparative welfare analysis between bidding and random roscas. Our work paints a more complete picture by giving quantified bounds on roscas' inefficiency, even in the face of heterogeneity of participants' values across agents and across time.

Our results apply techniques from auction theory and the price of anarchy literature. In particular, we make extensive use of {\em smoothness}, formalized in \citet{roughgarden2009intrinsic} and adapted to auctions in \citet{syrgkanis2013composable}. Smoothness is a sufficient condition for approximately-optimal equilibrium welfare, and is preserved by combination with other smooth mechanisms. In addition, smoothness-derived guarantees generalize beyond standard quasilinear, full-information settings to learning outcomes and Bayes-Nash equilibria, revenue \citet{hartline2014price}, large games \citet{feldman2016price}, risk-averse agents, \citet{kesselheim2018price}, and more. \citet{roughgarden2017price} give an excellent survey. Notably, because of the way payments are redistributed, smoothness is not sufficient for approximately optimal welfare in roscas. To derive our welfare results, we instead combine smoothness with extra analysis to control the impact of redistribution on welfare.

\section{Missing Proofs and Examples}

\subsection{Swap Rosca Example}

The example below shows that the bound from Theorem 1 is tight.

\begin{example}
Suppose we have three participants with value vectors: $\vecv_1=(0,0,0)$, $\vecv_2=(1,0,0)$, and $\vecv_3=(1,1,0)$. The optimal allocation allocates the pot to participant 2 in round 1, participant 3 in round 2, and participant 1 in round 3 for total welfare $\textsc{OPT}(\vecv)=2$. However, the initial allocation that allocates the pot to participant 1 in round 2, participant 2 in round 3, and participant 3 in round 1 is swap-stable and has welfare $1$. 
\end{example}

\subsection{Proof of Lemma \ref{lem:roundrobin}}

\lemroundrobin*

\begin{proof}

Given value profile $\vecv$, an optimal allocation is an assignment from participants $i$ to pots $t^*_i$.
For values $v$ and action profile $\veca$ in the sequential mechanism, we construct a deviation $a^*_i$ for each participant $i$ in the following way: participant $i$ simulates their equilibrium strategy up until round $t_i^* - 1$. Then, in round $t_i^*$, they play their smoothness deviation for $M$ on value profile $\hat \vecv$ and action $a_i^t$ for $i$, where $\hat v_i=v_i^t$ and $0$ for all other participants. They play $\bot$ in all subsequent rounds. Note that, in simulating $a_i$, participant $i$ may win in some round $\hat t_i\leq t^*_i$, in which case the mechanism excludes them from future rounds.

 Let $S_{\alg}$ denote the set of participants whose deviations cause them to win before $t_i^*$ and let $S_{\opt}$ denote the set of participants who do not win before $t^*_i$, hence are able to play their smoothness deviation in round $t^*_i$. For participants in $S_{\opt}$, the choice of smoothness deviation for round $t^*_i$ implies:
 \begin{align*}
      \vecv_i\cdot&\vecx_i(a_i^*,\veca_{-i})-\hat p_i(a_i^*,\veca_{-i}) \\
      &\geq \lambda \cdot v_i^{t_i^*}
      - \mu_1\sum\nolimits_k\hat p_k^{t_i^*}(\veca)-
      \mu_2 \sum\nolimits_kB_k^{t_i^*}(\veca),
 \end{align*}
where for any period $t$, $B_k^{t}(\veca)$ is participant $k$'s bid if they win in round $t$ and $0$ otherwise. 

 For participants in $S_{\alg}$, if $\hat{t}_i$ denotes the round they win in equilibrium, then $$\vecv_i\cdot \vecx_i(a_i^*,\veca_{-i})-\hat p_i(a_i^*,\veca_{-i})=v_i^{\hat t_i}-\hat p_i(\veca)\geq v_i^{t_i^*}-\hat p_i(\veca),$$ where the inequality follows from the fact that values are nonincreasing over time.

Summing over participants, we obtain:
\begin{align*}
    \sum_i&( \vecv_i\cdot\vecx_i(a_i^*,\veca_{-i})-\hat p_i(a_i^*,\veca_{-i}))\\
    &\geq \lambda\textsc{OPT}(\vecv)-\mu_1\sum\limits_{i \in S_{\opt}}\sum_k\hat p_k^{t_i^*}(\veca)-\sum\limits_{i \in S_{\alg}}\hat p_i(\veca)
    \\& -\mu_2\sum\limits_{i \in S_{\opt}}\sum_kB_k^{t_i^*}(\veca).
\end{align*}
The second and third terms on the righthand side of this inequality are each at least the total payments of the sequential mechanism. For the last term, we can write
\begin{equation*}
    \sum\limits_{i \in S_{\opt}}\sum_kB_k^{t_i^*}(\veca)\leq \sum_i\sum_t B_i^t(\veca)
     =\sum_i B_i(\veca),
\end{equation*}
where the first line follows from the fact that participants' bids are nonnegative and the second line from the fact that participants win in at most one round. We therefore obtain 
\begin{align*}
    \sum\nolimits_i&( \vecv_i\cdot\vecx_i(a_i^*,\veca_{-i})-\hat p_i(a_i^*,\veca_{-i}))\\
    &\geq \lambda\textsc{OPT}(\vecv)-(\mu_1+1)\sum\nolimits_i \hat p_i(\veca)-\mu_2\sum\nolimits_iB_i(\veca)
\end{align*}
\noindent giving us the desired inequality.
\end{proof}

The following example demonstrates the necessity of our no-overbidding assumption.
\begin{example}\label{ex:overbid}
Suppose participant $1$ has value $(10,0)$ and participant $2$ has value $(0,0)$. If participant $2$ bids $10$ in round $1$, then participant $1$ cannot win without incurring negative utility. participant $2$ thus wins for free and hence no rebates are distributed. The equilibrium welfare here is $0$, whereas the optimal welfare is $10$. 
\end{example}

\subsection{Proof of Lemma \ref{lem:overpay}}
\lemoverpay*

\begin{proof}
In what follows, we index bidders so that in the optimal allocation, participant $i$ wins in round $i$. 

We suppress dependence of payments and allocations on the profile of equilibrium bids when it is clear from context. We argue by strong induction on the number of rounds remaining. 

As our base case, note that participant $n$ could choose to bid $0$ and incur nonnegative utility from round $n$.\footnote{Note that this is different from their overall utility due to rebates in rounds $1$ through $n-1$. Also note that we can, in fact, say $\hat p_{n}^n=0$, but the weaker argument above yields cleaner indexing without changing the end constant.} Their continuation utility for playing according to their equilibrium strategy is $v_{n}^n-\hat p_{n}^n\geq0$, which yields the desired inequality for the base case.
	
	Now assume that for all rounds from $t+1$ onward, the desired inequality holds. participant $t$ could choose to bid $0$ in all rounds starting with $t$. The continuation utility from this deviation is at least $0$. Meanwhile, in equilibrium, $t$ is receiving $v_{t}^t-\hat p_{t}^t+\sum_{t'=t+1}^n\hat p_{t'}^{t'}/(n-1)\geq 0$ from future rounds. We therefore can upper bound $\hat p_{t}^t$ by
	\begin{align*}
 v_{t}^t+\tfrac{1}{n-1}\sum_{t'=t+1}^n\left [v_{t'}^{t'}+\tfrac{1}{n-1}\sum_{t''=t'+1}^nv_{t''}^{t''}\left(\tfrac{n}{n-1}\right)^{\tau}\right],
	\end{align*}
	
\noindent where $\tau = t''-t'-1$.  Collecting coefficients on $v_{t'}^{t'}$ gives us the following equivalent upper bound:
\begin{align*}
\hat p_{t}^t&\leq v_{t}^t+\tfrac{1}{n-1}\sum_{t'=t+1}^nv_{t'}^{t'}\left(1+\tfrac{1}{n-1}\sum_{j=0}^{t'-t-2}\left(\tfrac{n}{n-1}\right)^j\right)\\
&=v_{t}^t+\tfrac{1}{n-1}\sum_{t'=t+1}^nv_{t'}^{t'}\left(1+\tfrac{1}{n-1}\tfrac{1-(\tfrac{n}{n-1})^{t'-t-1}}{1-\tfrac{n}{n-1}}\right)\\
&=v_{t}^t+\tfrac{1}{n-1}\sum_{t'=t+1}^nv_{t'}^{t'}(\tfrac{n}{n-1})^{t'-t-1}.~~~\qedhere
\end{align*}
\end{proof}

\begin{proof}[Proof of Theorem~\ref{thm:firstprice}]
We first sum the bounds from Lemma~\ref{lem:overpay}:
\begin{align*}
	\sum_i \hat p_i(\veca) &= \sum_{t=1}^n \hat p_{t}^t(\veca) \\
	&\leq \sum_{t=1}^n\left[v_{t}^t+\tfrac{1}{n-1}\sum_{t'=t+1}^nv_{t'}^{t'}(\tfrac{n}{n-1})^{t'-t-1}\right]\\
	&=\sum_{t=1}^n \left(\tfrac{n}{n-1}\right)^{t-1}v_{t}^t\\
	&\leq e\sum_{t=1}^n v_{t}^t=e\sum_i \vecv_i\cdot \vecx_i(\veca).
\end{align*}
By Lemma~\ref{lem:roundrobin}, the first-price rosca is $(1-1/e,2,0)$-smooth before rebates. Considering the same deviation strategy with rebates, we obtain,
\begin{align*}
    \sum\nolimits_i u_i^{\vecv_i}(\veca)&\geq \sum\nolimits_iu_i^{\vecv_i}(a_i^*,\veca_{-i})\\
    &\geq \sum\nolimits_i( \vecv_i\cdot \vecx_i(a_i^*,\veca_{-i})-\hat p_i(a_i^*,\veca_{-i}))\\
    &\geq (1-1/e)\textsc{OPT}(\vecv)-2\sum\nolimits_i\hat p_i(\veca)\\
    &\geq (1-1/e)\textsc{OPT}(\vecv)-2e\sum\nolimits_i\vecv_i\cdot x_i(\veca) \\
\end{align*}
The result follows from noting that both $\sum_i u_i^{\vecv_i}(\veca)$ and $\sum_i \vecv_i\cdot \vecx_i(\veca)$ represent the equilibrium welfare.
\end{proof}

\section{Extension to Nonlinear Utilities}
\label{app:nonlinear}

In this appendix, we show how our auction results for up-front roscas and no-overbidding sequential roscas extend to the more general setting of nonlinear cost for money $C$. We consider the class of convex functions $C$ satisfying $C(0)=0$ and  $C'(x)\in[\alpha,\beta]$ on the range of relevant values for some $0<\alpha\leq\beta$, and we prove bounds which degrade linearly in the ratio $\beta/\alpha$.If $C$ is allowed to be an arbitrary convex function, then it could serve as a hard budget; auctions with hard budgets are known to have unbounded PoA \citep{syrgkanis2013composable,DP14}. Under this regime we can take advantage of two facts which follow immediately from basic calculus:

\begin{lemma}
Let $C$ be convex and increasing and satisfy $C'(x)\in[\alpha,\beta]$ and $C(0)=0$, with $0<\alpha\leq \beta$.

Then:
\begin{enumerate}
	\item For $x\geq 0$, $C(x)\in[\alpha x,\beta x]$.
	\item For $x\leq 0$, $C(x)\in[\beta x,\alpha x]$.
\end{enumerate}
\end{lemma}

As in the quasilinear case, we demonstrate our approach with up-front roscas, then present the more involved analysis of sequential roscas. Both the tradeoff step between utilities and payments and the upper bound on payments need to be adjusted to accommodate nonlinear utilities.

\subsection{Up-Front Roscas}

Under quasilinear utilities, it was unimportant to the analysis whether payments and rebates were made before round $1$ or distributed across time. When utilities are nonlinear, recall that an participant $i$'s disutility for payments is broken into the sum $\sum_t C(p_i^t)$. Hence, the timing of payments and rebates could change payoffs dramatically. Our analysis will be agnostic to these timing considerations. Whatever the timing, we continue to denote by $p_i^t(\vecb),\hat p_i^t(\vecb)$, and $\hat r_i^t(\vecb)$ the net payments, gross payments, and rebates, respectively, of participant $i$ at time $t$ under bid profile $\vecb$. Under any up-front rosca, participant $i$'s total payments are their bid. Hence, $\sum_t\hat p_i^t(\vecb)=b_i$ and $\sum_t\hat r_i^t(\vecb)=\sum_{j\neq i}b_j/(n-1)$.

Trading off payments and utilities becomes more complex. We extend the proof from \citet{syrgkanis2013composable}, parameterizing by a constant $\rho$:

\begin{lemma}\label{lem:ufsmoothconcave}
For any $\rho\in[0,2/\beta]$ and any Nash equilibrium $\vecb$ of an up-front rosca with values $\vecv$:
\begin{equation*}
	\sum\nolimits_i u_i^{\vecv_i}(\vecb)\geq (1-\tfrac{\rho\beta}{2})\textsc{OPT}(\vecv)-\tfrac{1}{\rho}\sum\nolimits_i\hat p_i(\vecb).
\end{equation*}
\end{lemma}
\begin{proof}
Pick a player $i$, and consider the deviation bid $b_i^*\sim U[0,\rho v_i^{t^*}]$, where $t^*$ denotes round $i$ is allocated in the optimal assignment. Let $\hat i$ denote the participant who wins $t^*$ under $\vecb$. Then we can lower bound the deviation utility $	u_i^{\vecv_i}(b_i^*,\vecb_{-i})$ as:
\begin{align*}
& \int_0^{\rho v_i^{t^*}}\big(v_i^{t^*}x_i^{t^*}(b_i^*,\vecb_{-i})-\sum_tC(p_i^t(b_i^*,\vecb_{-i}))\big) /\rho v_i^{t^*}\,db_i^*\\
&~~~~~~~~\geq \int_0^{\rho v_i^{t^*}}\big(v_i^{t^*}x_i^{t^*}(b_i^*,\vecb_{-i})-C(b_i^*)\big) /\rho v_i^{t^*}\,db_i^*\\
&~~~~~~~~ \geq\int_0^{\rho v_i^{t^*}}\big(v_i^{t^*}x_i^{t^*}(b_i^*,\vecb_{-i})-\beta b_i^*\big) /\rho v_i^{t^*}\,db_i^*\\
&~~~~~~~~= \int_{b_{\hat i}}^{\rho v_i^{t^*}}\tfrac{1}{\rho}\,db_i^*-\tfrac{\beta\rho v_i^{t^*}}{2}\\
&~~~~~~~~= v_i^{t^*}-\tfrac{b_{\hat i}}{\rho}-\tfrac{\beta\rho v_i^{t^*}}{2}
\end{align*}
Summing the deviation utilities across participants yields the desired bound.
\end{proof}
To upper bound the impact of rebates, we use:
\begin{lemma}\label{lem:ufoverbidconcave}
	Let $\vecb$ be a Nash equilibrium of an up-front rosca. Then for any participant $i$, $\hat p_i(\vecb)\leq \vecv_i\cdot\vecx_i(\vecb)/\alpha$.
\end{lemma}
\begin{proof}
	Assume for some $i$ that $\alpha \hat p_i(\vecb)> \vecv_i\cdot\vecx_i(\vecb)$. Then $i$ could improve their utility by bidding $0$, which, in an up-front rosca, does not change their rebates: $\hat r_i(0,\vecb_{-i})=\hat r_i(\vecb)> \vecv_i\cdot \vecx_i(\vecb)-\alpha \hat p_i(\vecb)+\hat r_i(\vecb)\geq \vecv_i\cdot \vecx_i(\vecb)-\sum\nolimits_t C(\hat p_i^t)+\hat r_i(\vecb) $.
\end{proof}

Combining Lemmas~\ref{lem:ufsmoothconcave} and \ref{lem:ufoverbidconcave} and choosing $\rho$ as follows:
\begin{equation*}
	\rho = \frac{\beta +\sqrt{\beta(2\alpha+\beta)}}{\beta(\alpha+\beta + \sqrt{\beta(2\alpha+\beta)})}
\end{equation*}
yields the following price of anarchy bound:
\begin{theorem}
	With cost for money $C$ satisfying $C'(x)\in[\alpha,\beta]$ every Nash equilibrium of an up-front rosca has PoA at most:
	\begin{equation*}
		\textsc{PoA}\leq \frac{\alpha+\beta+\sqrt{\beta(2\alpha+\beta)}}{\alpha}.
	\end{equation*}
\end{theorem}
Note that the bound degrades linearly with the ratio $\beta/\alpha$, as promised. Moreover, taking $\alpha=\beta = 1$ slightly improves on the bound from Theorem~\ref{thm:upfront}; several of the other constants in the paper can be improved through similar optimization of the bid deviations. We eschewed such optimization in favor of readability.

\subsection{Sequential Roscas}

We take a smoothness-based approach to sequential roscas. We will use the following generalization, adapted to non-linear utilities:
\begin{definition}
	Let $M$ be a sequential single-bid auction. We say $M$ is $(\lambda,\mu_1,\mu_2)$-smooth if for every value profile $\vecv$ and action profile $\veca$, there exists a randomized action $a_i^*(a_i,\vecv)$ for each $i$ such that:
	\begin{align*}
		&\sum\nolimits_iu_i^{\vecv_i}(a_i^*(a_i,\vecv),\veca_{-i}) \\
		&\,\,\,\geq  \lambda \textsc{OPT}(\vecv)-\mu_1\sum\nolimits_i \hat p_i(\veca)-\mu_2\sum\nolimits_i C(B_i(\veca)),
	\end{align*}
	where $B_i(\veca)$ is $i$'s bid in the round where they win, or $0$ if no such round exists.
\end{definition}

First-and second-price auctions are smooth, via extensions of the arguments from \citet{syrgkanis2013composable}.

\begin{lemma}
	The single-round second-price auction is $(1,0,1)$-smooth.
\end{lemma}
\begin{proof}
	Let $h$ be the index of the participant with the highest value. Have this bidder bid $C^{-1}(v_h)$, while the remaining participants bid $0$. Let $i^*$ denote the index of the highest bidder in equilibrium, and $\hat i$ the index of the highest bidder other than $h$. (We may have $i^*=\hat i$.) If $h$ wins bidding $C^{-1}(v_h)$,
	their utility is $v_h-C(a_{\hat i})=\text{OPT}-C(a_{\hat i})\geq \text{OPT}-C(a_{i^*})$. Otherwise, they lose and earn utility $0$, while $a_{\hat i}\geq C^{-1}(v_h)$. It follows that $\text{OPT}-C(a_{i^*})=v_h-C(a_{i^*})\leq v_h-C(a_{\hat i})\leq v_h-C(C^{-1}(v_h))= 0$. This implies the desired smoothness bound.
\end{proof}

\begin{lemma}
	The single-item first-price auction is $((1-1/e^\beta)\beta^{-1},1,0)$-smooth.
\end{lemma}
\begin{proof}
	The highest valued participant, say index $h$, can deviate to submitting a randomized bid $a_h^*$ drawn from the distribution with density function $f(x) = \tfrac{1}{v_h - \beta \cdot x}$ and support $[0, (1 - 1/e^{\beta} )v_h/\beta]$. The utility of the highest participant from this deviation is:
	\begin{align*}
		u^{v_h}_{h}(a_h^*,\veca_{-h}) &= \int_{\max_{i \neq h}a_i}^{\left(1 - \frac{1}{e^{\beta}}\right)\frac{v_h}{\beta}}(v_h - C(x))f(x) dx\\
		&\geq \int_{\max_{i \neq h}a_i}^{\left(1 - \frac{1}{e^{\beta}}\right)\frac{v_h}{\beta}}(v_h - \beta \cdot x)f(x) dx\\
		&= \left(1 - \frac{1}{e^{\beta}}\right)\frac{v_h}{\beta} - \max_i a_i.
	\end{align*}
	Since all other bidders can get utility at least $0$ by e.g. bidding $0$, the stated smoothness bound holds.
\end{proof}

Under the generalized definition of smoothness, the round-robin composition result, Lemma~\ref{lem:roundrobin} holds without modification. Hence, we obtain:

\begin{corollary}\label{cor:smoothconvex}
	First-price roscas are $((1-1/e^\beta)\beta^{-1},2,0)$-smooth, and second-price roscas are $(1,1,1)$-smooth.
\end{corollary}

Under non-linear utilities, the appropriate notion of overbidding states that no participant bids in a way that could yield negative utility. This upperbounds their bids by an amount that depends on $C$:

\begin{definition}\label{def:nobconvex}
	Action profile $\veca$ satisfies {\em no-overbidding} if $C(B_i(\veca))\leq \vecv_i\cdot \vecx_i(\veca)$ for every participant $i$.
\end{definition}

Under no-overbidding, values and payments cannot exceed utilities by too much, as the following two Lemmas state. 
\begin{lemma}\label{lem:vversusu}
	In any no-overbidding profile $\veca$ of sequential rosca based on a strongly IR single-item auction: $$\sum_i\vecv_i\cdot\vecx_i(\veca)\leq \tfrac{\beta}{\alpha}\sum_iu_i^{\vecv_i}(\veca).$$
\end{lemma}
\begin{proof}
	By definition, $\sum_iu_i^{\vecv_i}(\veca)=\sum_i\vecv_i\cdot\vecx_i(\veca)-\sum_i\sum_tC(p_i^t(\veca))$. The last term can be further decomposed as $\sum_i\sum_tC(p_i^t(\veca))=\sum_i\sum_tC(\hat p_i^t(\veca))+\sum_i\sum_tC(-\hat r_i^t(\veca))$. By the no-overbidding assumption, $\sum_i\vecv_i\cdot\vecx_i(\veca)\geq \sum_i\sum_tC(\hat p_i^t(\veca))$. Hence, the the ratio of values to utilities is upper bounded by the ratio of disutility from payments to utility from rebates. That is:
	\begin{equation*}
		\frac{\sum_i\vecv_i\cdot\vecx_i(\veca)}{\sum_iu_i^{\vecv_i}(\veca)}\leq\frac{\sum_i\sum_tC(\hat p_i^t(\veca))}{-\sum_i\sum_tC(-\hat r_i^t(\veca))}. 
	\end{equation*}
	The lemma then follows from noting:
	\begin{align*}
		\sum_i\sum_tC(\hat p_i^t(\veca))&\leq \beta\sum_i\sum_t\hat p_i^t(\veca)\\
		-\sum_i\sum_tC(-\hat r_i^t(\veca))&\geq \alpha\sum_i\sum_t\hat p_i^t(\veca).~~~\qedhere
	\end{align*}
\end{proof}

\begin{lemma}\label{lem:pversusu}
	In any no-overbidding profile $\veca$ of sequential rosca based on a strongly IR single-item auction: $\sum_i \hat p_i(\veca)=\sum_i \hat r_i(\veca)\leq \alpha^{-1}\sum_i u_i^{\vecv_i}(\veca)$.
\end{lemma}
\begin{proof}
	By no-overbidding, $\sum_t C(\hat p_i^t)\leq \vecv_i\cdot\vecx_i(\veca)$. We may therefore write:
	\begin{align*}
		\sum_i \hat p_i(\veca)&=\sum_i \hat r_i(\veca)\\
		&\leq \alpha^{-1}\sum_i\sum_t-C(-\hat r_i^t(\veca))\\
		&\leq \alpha^{-1}\Big(\sum_i \vecv_i\cdot \vecx_i(\veca)\\
		&~~~~-\sum_i\sum_tC(\hat p_i^t(\veca))-\sum_i\sum_tC(-\hat r_i^t(\veca))\Big)\\
		&=\alpha^{-1}\sum_i u_i^{\vecv_i}(\veca),
	\end{align*}
	where the first inequality follows from upper bounding $C$, and the second from no-overbidding.
\end{proof}

\begin{theorem}\label{thm:smoothroscaconvex}
	Let $M$ be a strongly IR, single-item auction that is $(\lambda,\mu_1,\mu_2)$-smooth, with $\lambda\leq1$. With quasilinear participants, every no-overbidding Nash equilibrium of the corresponding auction rosca with rebates has PoA at most $(1+\mu_1\alpha^{-1}+\mu_2\beta\alpha^{-1})/\lambda$.

\end{theorem}
\begin{proof}
	The theorem follows from the inequalities below, justified after their statement:
	\begin{align*}
		\lambda\textsc{OPT}&\leq\sum \nolimits_iu_i^{\vecv_i}(a_i^*(a_i,\vecv),\veca_{-i}) + \mu_1\sum \nolimits_i \hat p_i(\veca)\\
		&~~~~~~~~~+\mu_2\sum \nolimits_i \hat C(B_i(\veca))\\
		&\leq\sum \nolimits_iu_i^{\vecv_i}(a_i^*(a_i,\vecv),\veca_{-i}) + \tfrac{\mu_1}{\alpha}\sum \nolimits_i u_i^{\vecv_i}(\veca)\\
		&~~~~~~~~~+\mu_2\sum \nolimits_i \hat C(B_i(\veca))\\
		&\leq\sum \nolimits_iu_i^{\vecv_i}(a_i^*(a_i,\vecv),\veca_{-i}) + \tfrac{\mu_1}{\alpha}\sum \nolimits_i u_i^{\vecv_i}(\veca)\\
		&~~~~~~~~~+\mu_2\sum \nolimits_i \vecv_i\cdot\vecx_i(\veca)\\
		&\leq\sum \nolimits_iu_i^{\vecv_i}(a_i^*(a_i,\vecv),\veca_{-i}) + \tfrac{\mu_1}{\alpha}\sum \nolimits_i u_i^{\vecv_i}(\veca)\\
		&~~~~~~~~~+\tfrac{\mu_2\beta}{\alpha}\sum \nolimits_i u_i^{\vecv_i}(\veca).
	\end{align*}
The first inequality follows from smoothness. The second follows from Lemma~\ref{lem:pversusu}, and the fourth from Lemma~\ref{lem:vversusu}. The third inequality follows from no-overbidding. Applying best response yields the result.
\end{proof}
We obtain price of anarchy bounds for second- and first-price rocscas as corollaries. In particular, note that the price of anarchy is scale-invariant with respect to $C$. That is, for any $\gamma>0$ and any equilibrium $\veca$ with values $\vecv$, the same action profile $\veca$ is also an equilibrium for values $\hat \vecv=\gamma \vecv$ and cost function $\hat C(p) = \gamma C(p)$. Hence, any bound proved for a scaling $\gamma\beta,\gamma\alpha$ of $\beta,\alpha$ also holds for $\beta,\alpha$. Taking $\gamma=1/\beta$ implies the promised linear dependencies:

\begin{corollary}\label{thm:spaconvex}
	Any Nash equilibrium of the second-price rosca satisfying no-overbidding has PoA at most $1+2\beta/\alpha$.
\end{corollary}
\begin{corollary}\label{thm:fpaconvex}
	Any Nash equilibrium of the first-price rosca satisfying no-overbidding has PoA at most $(1+2\beta/\alpha)(1-1/e)$.
\end{corollary}

Note that our first-price result requires the no-overbidding assumption. We leave open whether the less restrictive analysis of Section~\ref{sec:overbidding} can be extended in some way to the nonlinear setting as well.

\section{Other Extensions for Auction Results}

The smoothness framework is known to be robust to variations in equilibrium assumptions, and we inherit this robustness to a significant degree. 
First, in roscas and especially sequential ones, it is reasonable to assume that participants may not best respond perfectly. A more natural notion may be some form of $\epsilon$-best response, where participants maximize their utility up to an additive $\epsilon$ error. All our results hold under this generalization, with the welfare guarantees similarly degrading by an additive, $O(n\epsilon)$ factor.
Second, we may also want to study roscas under incomplete information, with each participant's values being drawn according to a prior.
Smoothness-based welfare results typically extend to such settings, and ours do as well in large part. In particular, Theorems~\ref{thm:upfront} and \ref{thm:smoothrosca} both hold under any Bayes-Nash equilibrium where participants' value vectors are drawn independently of one another. For Theorem~\ref{thm:smoothrosca}, this can be derived directly by mimicking the proof of Theorem 4.3 in \citet{syrgkanis2013composable}. The extension to Theorem~\ref{thm:upfront} holds by additionally observing that in any Bayes-Nash equilibrium of up-front rosca, no participant has an incentive to overbid.

Unfortunately, our proof of Lemma~\ref{lem:overpay} seems to rely on the full-information assumption. We leave it as an open question whether a Lemma~\ref{lem:overpay} can be extended beyond full-information environments. Either way, among all auction settings, roscas, where participants are typically members of tight-knit communities, are maybe the best candidate for assuming full information.

\section{Experiment: CRRA Utilities (Extended)}
\label{app:expcrrasupp}

\begin{table}[!ht]
\centering
\caption{Value Profiles for 9-participant Roscas Used in CRRA Utilities Experiment}
\label{table:crraprofiles9}
\begin{tabular}{c} 
\textbf{Profiles (Values in Rounds 1-9)} \\ \hline
2 0 0 0 0 0 0 0 0 \\ \hline
2 2 2 2 2 0 0 0 0 \\ \hline
5 5 0 0 0 0 0 0 0 \\ \hline
5 5 5 5 5 5 0 0 0 \\ \hline
8 8 8 0 0 0 0 0 0 \\ \hline
8 8 8 8 8 8 8 0 0 \\ \hline
8 8 8 5 5 5 2 2 2 \\ \hline
8 8 6 6 4 4 2 2 0 \\ \hline
8 7 6 5 4 3 2 1 0 \\
\end{tabular}
\end{table}

The value profiles for in our CRRA Utilities experiment in Section~\ref{sec:expcrra} are provided in Table~\ref{table:crraprofiles9}.  Recall that these nine value profiles feature six with {\em cutoff values} of the form $v_i^t = \overline v$ for all $i\leq \hat t$ for some $\hat t$ and three participants with values that are roughly linearly decreasing in time.  The average maximum value among cutoff participants is $5$, which matched the average value for linearly decreasing values.

We repeat this experiment using a 30-participant rosca with qualitatively similar value profiles as those used in the 9-participant setting: the 30 value profiles feature 20 with {\em cutoff values} of the form $v_i^t = \overline v$ for all $i\leq \hat t$ for some $\hat t$ and 10 participants with values that are roughly linearly decreasing in time.  The average maximum value among cutoff participants is $15.5$, which is very close to the average value for linearly decreasing values ($15$).  Table~\ref{table:crra30} presents the results of our CRRA Utilities experiment from Section~\ref{sec:expcrra} repeated in a 30-participant rosca.  Once again, we observe the same trends as in the 9-participant rosca experiment, providing further evidence to strengthen the observations and claims in Section~\ref{sec:expcrra}.

\begin{table}[!ht]
\centering
\caption{Swap Rosca Performance Under Different CRRA Parameters (OPT = 551, random baseline ratio = 1.580)}
\label{table:crra30}
\begin{tabular}{cc|c|c|c|c|c} 
 & & \multicolumn{5}{c}{$\mathbf{W}$} \\
 & & $\mathbf{1}$ & $\mathbf{2}$ & $\mathbf{3}$& $\mathbf{4}$& $\mathbf{5}$ \\ \cline{2-7}
 \multirow{9}{*}{$\mathbf{a}$} & $\mathbf{0}$ & 1.036 & 1.036 & 1.036 & 1.036 & 1.036 \\ \cline{2-7}
 & $\mathbf{0.1}$ & 1.105 & 1.104 & 1.087 & 1.086 & 1.085 \\ \cline{2-7}
 & $\mathbf{0.2}$ & 1.105 & 1.104 & 1.088 & 1.088 & 1.086 \\ \cline{2-7}
 & $\mathbf{0.3}$ & 1.104 & 1.104 & 1.105 & 1.090 & 1.089 \\ \cline{2-7}
 & $\mathbf{0.5}$ & 1.105 & 1.104 & 1.104 & 1.104 & 1.104 \\ \cline{2-7}
 & $\mathbf{0.75}$ & 1.104 & 1.104 & 1.105 & 1.104 & 1.104 \\ \cline{2-7}
 & $\mathbf{1}$ & 1.104 & 1.105 & 1.105 & 1.104 & 1.105 \\ \cline{2-7}
 & $\mathbf{1.5}$ & 1.104 & 1.104 & 1.105 & 1.105 & 1.105 \\ \cline{2-7}
 & $\mathbf{2}$ & 1.104 & 1.104 & 1.105 & 1.104 & 1.104 \\
\end{tabular}
\end{table}

\section{Experiment: Distributional Diversity}
\label{app:experimentDist}

Tables~\ref{table:needs9}-~\ref{table:needs30} summarize our second set of experiments, which vary the composition of the population of participants in both $9$-participant and 30-participant roscas. More specifically, we consider seven different configurations of value profiles, described briefly below and provided for the 9-participant rosca in Table~\ref{table:distprofiles9} (qualitatively similar profiles were again used for the 30-participant rosca, as in our CRRA Utilities experiments). All participants have cutoff values, with participant $i$'s cutoff at $t=i$. We then vary (a) the distribution of magnitudes for participants' values and (b) the correlation of an participant's value with their cutoff round. Distributions labeled ``-dec'' have values which are negatively correlated with the cutoff round, and ``-inc'' have positively correlated values. The instance ``pointmass'' gives all participants constant value up to their cutoffs. The remaining instances can be described by the distribution of participants' constant values before their cutoffs: ``unif'' has a distribution uniform on 
$\{1,\ldots, 9\}$, ``pareto'' a pareto distribution, and ``unim'' a unimodal distribution with its mode at $4$. We further consider both quasilinear and CRRA utilities.

\begin{table}[!ht]
\centering
\caption{Value Profiles for 9-participant Roscas Used in the Distributional Diversity Experiment}
\label{table:distprofiles9}

\begin{subtable}{0.15\textwidth}
\centering
\begin{tabular}{c} 
0 0 0 0 0 0 0 0 0 \\ \hline
4 0 0 0 0 0 0 0 0 \\ \hline
4 4 0 0 0 0 0 0 0 \\ \hline
4 4 4 0 0 0 0 0 0 \\ \hline
4 4 4 4 0 0 0 0 0 \\ \hline
4 4 4 4 4 0 0 0 0 \\ \hline
4 4 4 4 4 4 0 0 0 \\ \hline
4 4 4 4 4 4 4 0 0 \\ \hline
4 4 4 4 4 4 4 4 0 \\
\end{tabular}
\caption{\textbf{pointmass}}
\end{subtable}
\hfill
\begin{subtable}{0.15\textwidth}
\centering
\begin{tabular}{c} 
9 0 0 0 0 0 0 0 0 \\ \hline
8 8 0 0 0 0 0 0 0 \\ \hline
7 7 7 0 0 0 0 0 0 \\ \hline
6 6 6 6 0 0 0 0 0 \\ \hline
5 5 5 5 5 0 0 0 0 \\ \hline
4 4 4 4 4 4 0 0 0 \\ \hline
3 3 3 3 3 3 3 0 0 \\ \hline
2 2 2 2 2 2 2 2 0 \\ \hline
1 1 1 1 1 1 1 1 1 \\
\end{tabular}
\caption{\textbf{unif-dec}}
\end{subtable}
\hfill
\begin{subtable}{0.15\textwidth}
\centering
\begin{tabular}{c} 
0 0 0 0 0 0 0 0 0 \\ \hline
1 1 0 0 0 0 0 0 0 \\ \hline
2 2 2 0 0 0 0 0 0 \\ \hline
3 3 3 3 0 0 0 0 0 \\ \hline
4 4 4 4 4 0 0 0 0 \\ \hline
5 5 5 5 5 5 0 0 0 \\ \hline
6 6 6 6 6 6 6 0 0 \\ \hline
7 7 7 7 7 7 7 7 0 \\ \hline
8 8 8 8 8 8 8 8 8 \\
\end{tabular}
\caption{\textbf{unif-inc}}
\end{subtable}

\begin{subtable}{0.5\textwidth}
\centering
\begin{tabular}{ccccccccc} 
12.7 & 0 & 0 & 0 & 0 & 0 & 0 & 0 & 0 \\ \hline
6.4 & 6.4 & 0 & 0 & 0 & 0 & 0 & 0 & 0 \\ \hline
4.2 & 4.2 & 4.2 & 0 & 0 & 0 & 0 & 0 & 0 \\ \hline
3.2 & 3.2 & 3.2 & 3.2 & 0 & 0 & 0 & 0 & 0 \\ \hline
2.5 & 2.5 & 2.5 & 2.5 & 2.5 & 0 & 0 & 0 & 0 \\ \hline
2.1 & 2.1 & 2.1 & 2.1 & 2.1 & 2.1 & 0 & 0 & 0 \\ \hline
1.8 & 1.8 & 1.8 & 1.8 & 1.8 & 1.8 & 1.8 & 0 & 0 \\ \hline
1.6 & 1.6 & 1.6 & 1.6 & 1.6 & 1.6 & 1.6 & 1.6 & 0 \\ \hline
1.4 & 1.4 & 1.4 & 1.4 & 1.4 & 1.4 & 1.4 & 1.4 & 1.4 \\
\end{tabular}
\caption{\textbf{pareto-dec}}
\end{subtable}

\begin{subtable}{0.5\textwidth}
\centering
\begin{tabular}{ccccccccc}
0 & 0 & 0 & 0 & 0 & 0 & 0 & 0 & 0 \\ \hline
1.6 & 0 & 0 & 0 & 0 & 0 & 0 & 0 & 0 \\ \hline
1.8 & 1.8 & 0 & 0 & 0 & 0 & 0 & 0 & 0 \\ \hline
2.1 & 2.1 & 2.1 & 0 & 0 & 0 & 0 & 0 & 0 \\ \hline
2.5 & 2.5 & 2.5 & 2.5 & 0 & 0 & 0 & 0 & 0 \\ \hline
3.2 & 3.2 & 3.2 & 3.2 & 3.2 & 0 & 0 & 0 & 0 \\ \hline
4.2 & 4.2 & 4.2 & 4.2 & 4.2 & 4.2 & 0 & 0 & 0 \\ \hline
6.4 & 6.4 & 6.4 & 6.4 & 6.4 & 6.4 & 6.4 & 0 & 0 \\ \hline
12.7 & 12.7 & 12.7 & 12.7 & 12.7 & 12.7 & 12.7 & 12.7 & 0 \\
\end{tabular}
\caption{\textbf{pareto-inc}}
\end{subtable}

\begin{subtable}{0.5\textwidth}
\centering
\begin{tabular}{ccccccccc} 
7.2 & 0 & 0 & 0 & 0 & 0 & 0 & 0 & 0 \\ \hline
5.6 & 5.6 & 0 & 0 & 0 & 0 & 0 & 0 & 0 \\ \hline
5.6 & 5.6 & 5.6 & 0 & 0 & 0 & 0 & 0 & 0 \\ \hline
4.0 & 4.0 & 4.0 & 4.0 & 0 & 0 & 0 & 0 & 0 \\ \hline
4.0 & 4.0 & 4.0 & 4.0 & 4.0 & 0 & 0 & 0 & 0 \\ \hline
4.0 & 4.0 & 4.0 & 4.0 & 4.0 & 4.0 & 0 & 0 & 0 \\ \hline
2.4 & 2.4 & 2.4 & 2.4 & 2.4 & 2.4 & 2.4 & 0 & 0 \\ \hline
2.4 & 2.4 & 2.4 & 2.4 & 2.4 & 2.4 & 2.4 & 2.4 & 0 \\ \hline
0.8 & 0.8 & 0.8 & 0.8 & 0.8 & 0.8 & 0.8 & 0.8 & 0.8 \\
\end{tabular}
\caption{\textbf{unim-dec}}
\end{subtable}

\begin{subtable}{0.5\textwidth}
\centering
\begin{tabular}{ccccccccc}
0 & 0 & 0 & 0 & 0 & 0 & 0 & 0 & 0 \\ \hline
2.4 & 0 & 0 & 0 & 0 & 0 & 0 & 0 & 0 \\ \hline
2.4 & 2.4 & 0 & 0 & 0 & 0 & 0 & 0 & 0 \\ \hline
4.0 & 4.0 & 4.0 & 0 & 0 & 0 & 0 & 0 & 0 \\ \hline
4.0 & 4.0 & 4.0 & 4.0 & 0 & 0 & 0 & 0 & 0 \\ \hline
4.0 & 4.0 & 4.0 & 4.0 & 4.0 & 0 & 0 & 0 & 0 \\ \hline
5.6 & 5.6 & 5.6 & 5.6 & 5.6 & 5.6 & 0 & 0 & 0 \\ \hline
5.6 & 5.6 & 5.6 & 5.6 & 5.6 & 5.6 & 5.6 & 0 & 0 \\ \hline
7.2 & 7.2 & 7.2 & 7.2 & 7.2 & 7.2 & 7.2 & 7.2 & 0 \\

\end{tabular}
\caption{\textbf{unim-inc}}
\end{subtable}
\end{table}

We first note that as in the previous experiment, swap roscas provided across-the-board improvements over both the worst-case ratios and those of random allocation, and that under CRRA utilities, the social welfare degraded slightly compared to quasilinear. Beyond these observations, we can glean further insights from pairwise comparisons. Distributions labeled ``-dec'' represent settings where higher-valued participants have more urgent need for allocation. Both random allocation and swap roscas (under CRRA utilities) perform poorly on these instances compared to their ``-inc'' counterparts, whereas swap roscas (under quasilinear utilities) achieved better (lower) performance ratios.  Swap rosca welfares were roughly double those of the random allocations (i.e., swap roscas achieved performance ratios that were half those of random allocations).  A second informative set of comparisons is between ``pointmass'' and the ``unim'' distrbutions. Under both quasilinear and CRRA utilities, the approximation ratio seems to be driven by the correlation between urgency and value much more than the level of inequality in needs: if the latter was the main concern, both  ``unim'' distributions would see worse performance. Swap roscas seem well-equipped to coordinate allocation among heterogeneous participants.

\begin{table}
\centering
\caption{Swap Rosca Performance Under Diverse Value Distributions (9 participants). CRRA parameter values $W=4$, $a=.5$}
\label{table:needs9}
\begin{tabular}{c|c|c|c|c} 
 \textbf{Profile} & \textbf{OPT} & \textbf{Random} & \textbf{Quasilinear} & \textbf{CRRA} \\ \cline{1-5}
\textbf{pointmass} & 32 & 2.002 & 1.220 & 1.220 \\ \cline{1-5}
\textbf{unif-dec} & 45 & 2.452 & 1.000 & 1.327 \\ \cline{1-5}
\textbf{unif-inc} & 36 & 1.351 & 1.030 & 1.031 \\ \cline{1-5}
\textbf{pareto-dec} & 35.994 & 2.828 & 1.000 & 1.451 \\ \cline{1-5}
\textbf{pareto-inc} & 34.580 & 1.485 & 1.096 & 1.106 \\ \cline{1-5}
\textbf{unim-dec} & 36 & 2.337 & 1.116 & 1.310 \\ \cline{1-5}
\textbf{unim-inc} & 35.2 & 1.708 & 1.138 & 1.141 \\
\end{tabular}
\end{table}

\begin{table}
\centering
\caption{Swap Rosca Performance Under Diverse Value Distributions (30 participants). CRRA parameter values $W=4$, $a=.5$}
\label{table:needs30}
\begin{tabular}{c|c|c|c|c} 
 \textbf{Profile} & \textbf{OPT} & \textbf{Random} & \textbf{Quasilinear} & \textbf{CRRA} \\ \cline{1-5}
\textbf{pointmass} & 116 & 2.002 & 1.218 & 1.218 \\ \cline{1-5}
\textbf{unif-dec} & 465 & 2.819 & 1.000 & 1.385 \\ \cline{1-5}
\textbf{unif-inc} & 435 & 1.452 & 1.043 & 1.058 \\ \cline{1-5}
\textbf{pareto-dec} & 399.499 & 4.000 & 1.000 & 1.757 \\ \cline{1-5}
\textbf{pareto-inc} & 396.165 & 1.323 & 1.066 & 1.061 \\ \cline{1-5}
\textbf{unim-dec} & 400 & 2.553 & 1.098 & 1.347 \\ \cline{1-5}
\textbf{unim-inc} & 398.667 & 1.636 & 1.114 & 1.114 \\
\end{tabular}
\end{table}

\section{Simulation Code}

All code used for simulations can be found at:
\center{
{\tt\href{https://github.com/cikeokwu/swap_rosca_sims}{github.com/cikeokwu/swap\_rosca\_sims}}
}

\end{document}